\documentclass[12pt, onecolumn]{IEEEtran}

\IEEEoverridecommandlockouts
\usepackage{amsthm, amssymb}

\def\dint{{\rm \ d}}

\newtheoremstyle{slplain}
  {3pt}
  {3pt}
  {\slshape}
  {}
  {\bfseries}
  {.}%
  { }
  {}

\theoremstyle{slplain}
\newtheorem{thm}{Theorem}
\newtheorem{cor}{Corollary}

\usepackage{cite}
\usepackage{amsfonts}
\usepackage{multicol,multienum}
\usepackage{subfigure}
\usepackage{multirow}

\ifCLASSINFOpdf
  \usepackage[pdftex]{graphicx}
  \graphicspath{{../pdf/}{../jpeg/}}
  \DeclareGraphicsExtensions{.pdf,.jpeg,.png}
\else
  \usepackage{float}
  \usepackage[dvips]{graphicx}
  \graphicspath{{../eps/}}
  \DeclareGraphicsExtensions{.eps}
\fi

\usepackage[cmex10]{amsmath}
\usepackage{algorithm}
\usepackage{algorithmic}
\usepackage{array}
\usepackage{url}
\usepackage{setspace}

\begin{document}

\title{Connectivity of Millimeter Wave Networks\\ with Multi-hop Relaying}

\author{
\IEEEauthorblockA{Xingqin Lin and Jeffrey G. Andrews}
\thanks{Xingqin Lin and Jeffrey G. Andrews are with the Department of Electrical $\&$ Computer Engineering, The University of Texas at Austin, TX, USA. (Email: xlin@utexas.edu, jandrews@ece.utexas.edu). %Date revised: \today. }
}
}

\maketitle

\begin{abstract}
This paper presents a novel stochastic geometry approach to the connectivity of milimeter wave (mmWave) networks with multi-hop relaying. The random positions and shapes of obstacles in the radio environment are modeled as a Boolean model, whose germs are distributed according to a Poisson point process and grains are random rectangles.  The derived analytical results shed light on how the connectivity of mmWave networks depends on key system parameters such as the density and size of obstacles as well as \textit{relaying route window} -- the range of distances in which routing relays are selected. We find that multi-hop relaying can greatly improve the connectivity versus single hop mmWave transmission. We show that to obtain near-optimal connectivity the relaying route window should be about the size of the obstacles.
\end{abstract}

%\begin{IEEEkeywords}
%Device-to-device relaying, millimeter wave, connectivity, relay selection, stochastic geometry.
%\end{IEEEkeywords}

%\IEEEpeerreviewmaketitle

\section{Introduction}

The electromagnetic spectrum from 30 to 300 GHz is known as the milimeter wave (mmWave) band.
% -- a name due to the fact that the radio wavelengths of this band range from one to ten millimeters. 
Traditionally, the mmWave band has been commonly used in radio astronomy, remote sensing and radar applications \cite{daniels2014millimeter}. More recently, there has been a surge of interest in utilizing the large chunks of spectrum resources in the mmWave band for telecommunications, mainly because the current spectrum resources below mmWave frequencies are not sufficient to meet the exponentially growing demand for wireless traffic \cite{rangan2014millimeter}.

Despite the great potential of mmWave band for telecommunications, making it work is challenging since mmWave propagation characteristics are quite different from those of the spectrum below 5 GHz \cite{pi2011introduction, rangan2014millimeter}. In particular, mmWaves are susceptible to blockage as they have limited diffraction ability and cannot penetrate many solid materials. Therefore, the modeling and analysis of blockage are critical in mmWave system design. The work \cite{bai2013analysis} proposes to use a random Boolean model \cite{baccelli2009stochastic} for the spatially distributed blocking obstacles, i.e., the positions of the obstacles are modeled by a Poisson point process (PPP) and the shapes of the obstacles are random rectangles. The analysis however is restricted to single-hop transmission.

%How to overcome these challenges as well as many others is currently being investigated in both academia and industry \cite{pi2011introduction, rangan2014millimeter}.

%Despite the great potentials of mmWave band for telecommunications, making it work is challenging since the propagation characteristics of mmWave band is quite different from those of microwave band \cite{rangan2014millimeter}. On the one hand, radio waves in mmWave band suffer from high free space loss,  atmospheric absorption, and higher noise floor associated with larger bandwidth, limiting the range of mmWave transmission. On the other hand, radio waves in mmWave band are susceptible to blockage as they have limited diffraction ability and cannot penetrate most solid materials. How to overcome these challenges as well as many others is currently being investigated in both academia and industry \cite{pi2011introduction, rangan2014millimeter}.

%, as well as standardization bodies include IEEE 802.11ad \cite{vaughan2010gigabit} and 3GPP LTE-A.

Multi-hop relaying has the potential to extend mmWave coverage and network connectivity \cite{rangan2014millimeter, singh2009blockage}. For example, suppose a source would like to communicate to a destination using mmWave transmission. The communication would be unsuccessful if the mmWave signal arriving at the destination is not strong enough due to the various blockages in the environment. Instead, relaying via e.g., device-to-device communication \cite{lin2013overview} may help the mmWave signal turn around obstacles and set up the connection provided there exists a feasible path \cite{singh2009blockage}. The diffraction model proposed in \cite{singh2009blockage}  is based on ray tracing and is of high complexity when applied to the analysis and simulation of large mmWave networks.

In this paper, we take a stochastic geometry approach and model the obstacles by a Boolean model, as done in \cite{bai2013analysis}. We present a novel analysis on the fundamental performance of the connectivity of mmWave networks with multi-hop relaying. 
%The derived analytical results shed light on how to increase the connectivity of mmWave networks.
The derived results show that multi-hop relaying can greatly improve the connectivity versus single hop mmWave transmission. To obtain near-optimal connectivity, we find that the range of distances in which routing relays are selected should be about the size of the obstacles.

\section{Models and Assumptions}

Consider an arbitrary source-destination pair whose source and destination are respectively located at $x$ and $y$. The source $x$ communicates with the destination $y$ via directional mmWave signals. We assume that the mmWave signals cannot penetrate obstacles. If the LOS path between $x$ and $y$ is blocked, the directional mmWave signals are routed by intermediate relays to turn around obstacles. We assume that the helping relays are selected from the pool of potential relays located in the rectangular strip $R_\kappa (x, y)$, which has width $\kappa $ and is symmetric with respect to $x\to y$. See Fig. \ref{fig:1} for an illustration. The physical meaning of $\kappa$ is the range of area in which routing relays are selected. Intuitively, the larger the width $\kappa$, the higher the routing complexity.
%Restricting the relay selection to the rectangular strip $R_\kappa  (x, y)$ helps reduce complexity as the routing decision can be made only based on local information about the relays located in $R_\kappa  (x, y)$. 
We term the width $\kappa$ the \textit{relaying route window}.

Note that even with routing using intermediate relays in $R_\kappa  (x, y)$, it is still possible that no path exists which can route  the directional mmWave signals from the source $x$ to the destination $y$ due to the blockage of obstacles, as in Fig. \ref{fig:1}(b). In this paper, we explore how the choice of relaying route window $\kappa $ affects the connectivity of a typical source-destination pair. Intuitively, we can select a larger $\kappa $ to make the blockage probability smaller because a larger $\kappa $ results in a higher probability of finding a path from the source $x$ to the destination $y$ in the presence of obstacles. 

\begin{figure}
\centering
\includegraphics[width=8cm]{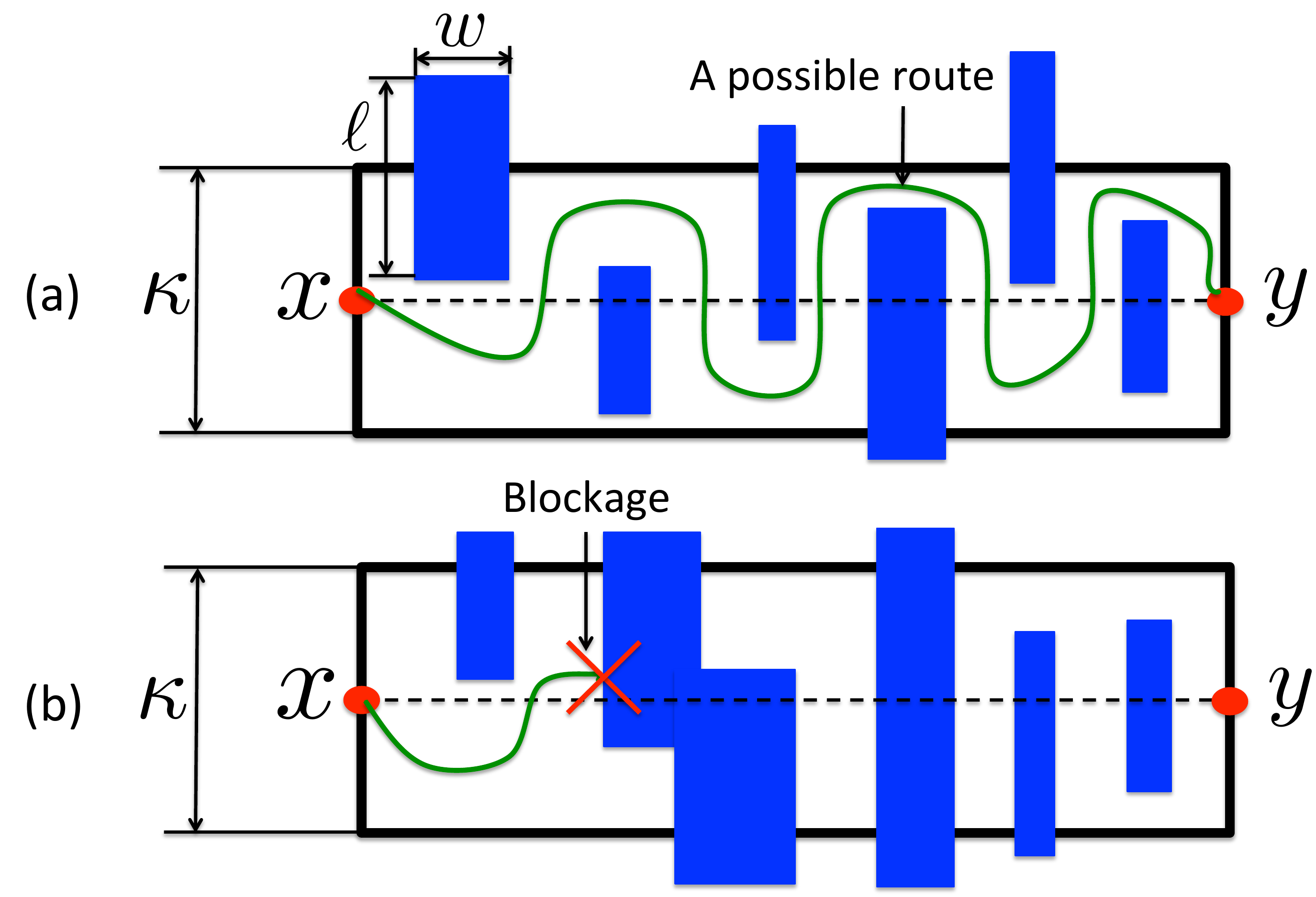}
%\caption{Two sample realizations of the Boolean blockage model. In the upper plot, it is possible to route the directional mmWave signal from the source $x$ to the destination $y$ with relaying. But in the  bottom plot such routing is not possible due to the blockage.}
\caption{Two sample realizations of the Boolean blockage model. The source $x$ and destination $y$ are connected in (a) but disconnected in (b).}
\label{fig:1}
\end{figure}

%As in \cite{baek2007spatial}, 
We consider an \textit{extreme} fluid model in which the density of potential relays in $R_\kappa  (x, y)$ tends to infinity. Investigating this asymptotic scenario will reveal the fundamental limit on the connectivity. The route selection can be arbitrary. Further, we model the obstacles by a Boolean model \cite{baccelli2009stochastic}. In this model, the centroid points of the obstacles are randomly distributed according to a homogeneous PPP $\Phi_o$ with intensity $\lambda_o$. The shapes of the obstacles are modeled by a sequence of i.i.d. rectangles $\{\Xi (W_{x_i}, L_{x_i})\}_{x_i \in \Phi_0}$ with randomly distributed widths $\{W_{x_i}\}$ and lengths $\{L_{x_i}\}$. We use $W$ and $L$ to denote the generic random variables for  $\{W_{x_i}\}$ and $\{L_{x_i}\}$, respectively, and assume that both $W$ and $L$ have finite supports, i.e., $W \in [w_{\min}, w_{\max})$ and $L \in [\ell_{\min}, \ell_{\max})$, but $L$ does not have to be larger than $W$. The finite support assumption is reasonable because in practice the obstacles are always of finite size.

Each rectangle $\Xi (W_{x_i}, L_{x_i})$ is positioned perpendicular to the LOS path $x\to y$ between the source $x$ and the destination $y$. Generalizing the model to allow each rectangle to have a random direction is possible but at the cost of notational complexity.  More formally, the Boolean obstacle model is driven by the following independently marked PPP:
\begin{align}
\tilde{\Phi}_o = \sum_{x_i \in \Phi_o} \delta_{\left(x_i, \Xi(W_{x_i}, L_{x_i}) \right) } ,
\end{align} 
where $\delta_{(\cdot)}$ denotes the Dirac measure, the centroid points $\{x_i\}$ are PPP distributed  \textit{germs}, and the marks $\{\Xi(W_{x_i}, L_{x_i})\}$ are the corresponding  i.i.d. rectangular \textit{grains}. 
%For technical reasons, we assume that the expected area of the rectangular grain is finite, i.e., $\mathbb{E} [ W \cdot L] <\infty$. 
The associated Boolean obstacle model is given by
\begin{align}
\Xi_{\textrm{obstacles}} =\bigcup_{x_i \in \Phi_o} \left(x_i + \Xi(W_{x_i}, L_{x_i}) \right) ,
\end{align} 
which is a random geometric object.

\section{Connectivity Analysis}

We are interested in finding the probability that the source $x$ and destination $y$ are connected, i.e., there exists at least one path that may be used to route the directional mmWave signals from the source $x$ to the destination $y$ with the help of the potential relays located in $R_\kappa  (x, y)$ if needed. This problem is demonstrated in Fig. \ref{fig:1}. We denote by $\mathcal{E}$ the connectivity event.
%Equivalently, we may think of this problem as follows. Imagine the rectangular strip $R_\kappa  (x, y)$ as a container. The container has randomly distributed obstacles; these obstacles may or may not be fully contained in the container. Now  we inject water into the container at position $x$. The water can only flow in the container and cannot penetrate the obstacles if any. Then what is the probability that the water percolates the container to wet the point $y$? 
The exact analysis of the connection probability turns out quite challenging. As such, we derive an upper bound given in the following theorem.

\begin{thm}
Assume w.l.o.g. $x=(0,0)$ and $y=(d,0)$. The probability of the connectivity event that there exists a path in $R_\kappa  (x, y)$ that may be used to route the directional mmWave signals from the source $x$ to the destination $y$ is upper bounded as
\begin{align}
\mathbb{P} (\mathcal{E}) 
&\leq \mathbb{P}^{(ub)} (\mathcal{E})  =  \exp \left( -\lambda_o \int \int A_{\kappa, d} (w,\ell) \dint F_W (w) \dint F_L (\ell)    \right),
\label{eq:1}
\end{align}
where 
\[  A_{\kappa, d} (w,\ell)  = \left\{ \begin{array}{ll}
         ( w +d )\ell & \mbox{if $w \geq d$};\\
         2 w  \ell & \mbox{if $w < d$ and $\ell < \kappa$};\\
         (\ell - \kappa) (d - w) + 2 w \ell & \mbox{if $w < d$
         and $\ell \geq \kappa$}.\end{array} \right. \] 
\label{thm:1}
\end{thm}
\begin{proof}
We defer the proof to Section \ref{proof:thm:1}.
\end{proof}

In the Boolean obstacle model $\Xi_{\textrm{obstacles}}$, different obstacles may intersect and jointly form a ``giant'' obstacle that blocks any path in $R_\kappa  (x, y)$ connecting the source $x$ to the destination $y$. This delicate correlation among obstacles in the blockage makes the analysis of the exact connection probability intractable. Instead, the upper bound $\mathbb{P}^{(ub)} (\mathcal{E})$ in (\ref{eq:1}) is derived by only counting blockage events caused by a single obstacle. 
Neglecting the correlation leads to an overestimated connection probability. 
Clearly, $\mathbb{P}^{(ub)} (\mathcal{E})$ is tight only if the correlation is small -- which holds if the obstacles are sparsely distributed, and so most blockages are indeed caused by a single obstacle.

We note that as $\kappa \to 0$, the upper bound becomes exact, i.e.,
\begin{align}
\lim_{\kappa \to 0} \mathbb{P} (\mathcal{E})  = \lim_{\kappa \to 0} \mathbb{P}^{(ub)} (\mathcal{E}).
\end{align}
This is because ``giant'' obstacles do not lead to more blockage as $\kappa \to 0$. Note that $\kappa \to 0$ also implies that we do not rely on relaying but simply use the direct LOS path from the source $x$ to the destination $y$. In other words, as $\kappa \to 0$, the connection probability equals the probability of the existence of the LOS path from the source to the destination. We formalize this observation in the following Corollary \ref{cor:1}, which shows that the LOS probability in the Boolean obstacle model exponentially decreases with the source-destination distance, the density of obstacles, and the mean length and width of the rectangular obstacles.
\begin{cor}
The direct LOS path from the source $x=(0,0)$ to the destination $y=(d,0)$ exists with probability
\begin{align}
p_{LOS} \triangleq \lim_{\kappa \to 0} \mathbb{P} (\mathcal{E})   = \exp \left( - \lambda_o \mathbb{E}[L] ( d + \mathbb{E}[W] )    \right) .
\end{align}
\label{eq:3}
\label{cor:1}
\end{cor}
\vspace{-1em}

To gain further insight, let us assume that the rectangular obstacles do not have random widths and lengths. Then Theorem \ref{thm:1} reduces to the following Corollary \ref{cor:2}.  
\begin{cor}
Assume that $W\equiv w$ and $L \equiv \ell$. The probability that the source $x=(0,0)$ and the destination $y=(d,0)$ are connected can be upper bounded by
\[  \mathbb{P}^{(ub)} (\mathcal{E})  = \left\{ \begin{array}{ll}
         \exp ( -\lambda_o ( w +d )\ell  ) & \mbox{if $w \geq d$};\\
         \exp (-\lambda_o \cdot 2 w  \ell ) & \mbox{if $w < d$ and $\ell < \kappa$};\\
          \exp (-\lambda_o ( (\ell - \kappa) (d - w) + 2 w \ell ) ) & \mbox{if $w < d$
         and $\ell \geq \kappa$}.\end{array} \right. \] 
\label{cor:2}
\end{cor}
\vspace{-1em}
Two remarks on Corollary \ref{cor:2} are in order.

\textbf{Remark 1.} When the width of a typical obstacle is greater than or equal to the source-destination distance, i.e. $w \geq d$, the upper bound on the connection probability equals $e^{-\lambda_o ( w +d )\ell  }$, which is independent of relaying route window $\kappa$. This is because when $w \geq d$ the blockage typically happens when either the source or the destination is located indoor. In this case, relaying is not very helpful in increasing the source-destination connectivity.

\textbf{Remark 2.}  When the width of a typical obstacle is smaller than the source-destination distance, i.e. $w < d$, relaying becomes more helpful in increasing the source-destination connectivity. In particular, with a fixed rectangular length $\ell$,  $\mathbb{P}^{(ub)} (\mathcal{E})$ decreases when $\kappa \in [0,\ell)$ and then stays constant when $\kappa \in [\ell,\infty)$. This observation gives a useful design rule: \textit{Set the relaying route window $\kappa$ equal to (in an order of magnitude sense) the obstacle length $\ell$.} Choosing a larger $\kappa$ beyond the obstacle length $\ell$ may further increase the connection probability. But the gain is expected to be marginal because the connection probability is always upper bounded by $e^{-\lambda_o \cdot 2 w  \ell }$ when $\kappa \geq \ell$.

The last remark gives a rule of thumb on selecting the relaying route window $\kappa$ when the rectangular obstacles have a common fixed width and length. The following Corollary \ref{cor:3} generalizes the result to the more general case where the rectangular obstacles may have random sizes. 
\begin{cor}
The relaying route window $\kappa^\star = \ell_{\max}$ maximizes the upper bound  $\mathbb{P}^{(ub)} (\mathcal{E})$ of connection probability, and the attained $\mathbb{P}^{(ub)} (\mathcal{E})$ is given by
\begin{align}
\mathbb{P}^{(ub)} (\mathcal{E})  =  
& \exp \left( -\lambda_o \int \ell \left( \int_{w\geq d} (w+d)   \dint F_W (w) +  \int_{w < d} 2w \dint F_W (w) \right)   \dint  F_L (\ell)   \right) .
\end{align}
\label{cor:3}
\end{cor}
\vspace{-2em}
%\begin{proof}
%From Theorem \ref{thm:1}, we have
%\begin{align}
%&\mathbb{P}^{(ub)} (\mathcal{E})  =  \textrm{const} \cdot  e^{-\lambda_o \int_{w<d} \int_{\ell \geq \kappa} (\ell - \kappa) (d-w) \dint F_L(\ell) \dint F_W (w) } \notag \\
%&=   \textrm{const} \cdot  e^{-\lambda_o \int_{w<d} (d-w) \dint F_W (w) \cdot  \int_{\ell \geq \kappa} (\ell-\kappa) \dint F_L(\ell)  },
%\label{eq:2}
%\end{align}
%where const is positive and is independent of relaying route window $\kappa$. Note that $\lambda_o \int_{w<d} (d-w) \dint F_W (w) \geq 0$. Thus, maximizing $\mathbb{P}^{(ub)} (\mathcal{E})$ is equivalent to minimizing $g(\kappa) \triangleq \int_{\ell \geq \kappa} (\ell - \kappa) \dint F_L (\ell)$. The derivative of $g(\kappa)$ with respect to $\kappa$  is  
%\begin{align}
%\frac{\dint g(\kappa)}{\dint \kappa} = - \int_{\ell \geq \kappa} \dint F_L (\ell) \leq 0 .
%\end{align}
%It follows that $g(\kappa)$ decreases with $\kappa$ and therefore $\kappa^\star = \ell_{\max}$. Plugging $\kappa^\star = \ell_{\max}$ into (\ref{eq:1}) yields the expression (\ref{eq:2}) for the corresponding upper bound.
%\end{proof}
The proof of Corollary \ref{cor:3} is omitted due to space limits.
Corollary \ref{cor:3} validates our intuition: A large relaying route window $\kappa$ is beneficial for connectivity. In particular, we should set $\kappa$ equal to the maximum obstacle length $\ell_{\max}$. As in the case with constant rectangular length, choosing a larger $\kappa$ beyond the maximum obstacle length $\ell_{\max}$ may further increase the actual connection probability. But the gain is expected to be small. This implies that in practice it suffices to perform neighbor discovery in routing in a range slightly larger than $\ell_{\max}$.

The connection probability characterized in Theorem \ref{thm:1} does not make any assumption on the locations of the source and destination relative to the Boolean obstacle model $\Xi_{\textrm{obstacles}}$. Note that when $x \in \Xi_{\textrm{obstacles}}$, the physical meaning is that $x$ is an indoor transmitter; otherwise, $x$ is an outdoor transmitter. So Theorem \ref{thm:1} characterizes the connectivity of a typical communication pair in which the source and destination can be either indoor or outdoor. For a typical cellular pair consisting of a base station (BS) and a mobile, the BS is usually located outdoor while the mobile can be either indoor or outdoor. Then conditioned on the source (resp. destination) being outdoor, what is the connection probability of the source-destination pair? Further, what is the connection probability if both the source and destination are outdoor? We answer these questions in Theorem \ref{pro:1}.

\begin{thm}
Conditioned on that $x \notin \Xi_{\textrm{obstacles}}$, i.e., $x$ is outdoor, the probability that the source $x=(0,0)$ and the destination $y=(d,0)$ are connected can be upper bounded by
\begin{align}
 \mathbb{P}^{(ub)} ( \mathcal{E}  | x \textrm{ is outdoor} )  &=  \exp \left( \lambda_o \mathbb{E}[W] \mathbb{E}[L] \right)  \cdot  \exp \left( -\lambda_o \int \int A_{\kappa, d} (w,\ell) \dint F_W (w) \dint F_L (\ell)    \right) .
\end{align}
Conditioned on that $\{x,y\} \notin \Xi_{\textrm{obstacles}}$, i.e., both $x$ and $y$ are outdoor, the probability that the source $x=(0,0)$ and the destination $y=(d,0)$ are connected can be upper bounded by
\begin{align}
&\mathbb{P}^{(ub)} ( \mathcal{E} | x \textrm{ and } y \textrm{ are outdoor} )   =  \exp \left( 2\lambda_o \mathbb{E}[W] \mathbb{E}[L] \right) \times
\notag \\  
&\exp \left( -\lambda_o \int \int ( A_{\kappa, d} (w,\ell) + (w-d)\ell \cdot \mathbb{I} (w \geq d) ) \dint F_W (w) \dint F_L (\ell)    \right).
\end{align}
\label{pro:1}
\end{thm}
\vspace{-2em}

The proof of Theorem \ref{pro:1} is similar to that of Theorem \ref{thm:1} and thus is omitted for brevity. From Theorems \ref{thm:1} and \ref{pro:1}, we can see that
\begin{align}
\mathbb{P}^{(ub)} ( \mathcal{E}  )  & \leq \mathbb{P}^{(ub)} ( \mathcal{E} | x \textrm{ is outdoor} )  \leq \mathbb{P}^{(ub)} (\mathcal{E} | x \textrm{ and } y \textrm{ are outdoor} ) .
\end{align}
The above relations are intuitive because conditioning on the source, or the destination, or both being outdoors reduces the blockage probability and thus increases the connectivity.

Before ending this section, we give a numerical example in Fig. \ref{fig:3}. It can be seen that multi-hop relaying can greatly improve the connectivity of mmWave networks. Further, setting the relaying route window $\kappa$ to just $20$m provides nearly optimal performance. Making $\kappa$ much larger (i.e., $100$m) provides marginal gains. These observations confirm our insights from the analytical results.

\begin{figure}
\centering
\includegraphics[width=8.5cm]{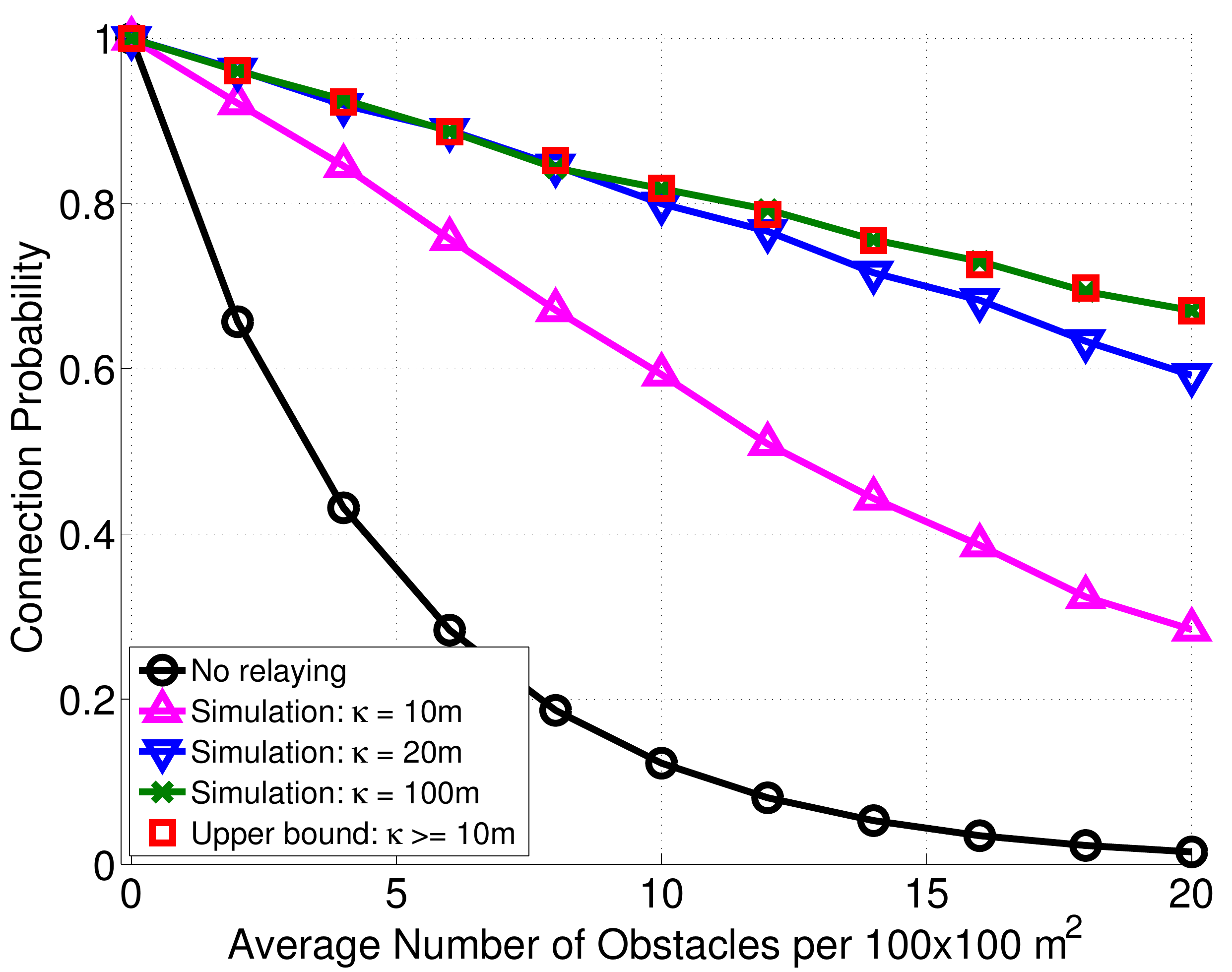}
\caption{Connectivity improvement of multi-hop relaying versus single hop mmWave transmission. Parameters: $d=200$m, $w=\ell=10$m. }
\label{fig:3}
\end{figure}

\section{Proof of Theorem \ref{thm:1}}
\label{proof:thm:1}

In this section, we prove Theorem \ref{thm:1}. To begin with, we partition the width interval $[w_{\min}, w_{\max})$ of the  rectangular obstacles as $w_{\min}=w_0<w_1<...<w_m=w_{\max}$. Let $\Delta w_i = w_i - w_{i-1}, i=1,...,m,$ be the length of the $i$-th subinterval $[w_{i-1}, w_i)$. Similarly, we partition the length interval $[\ell_{\min}, \ell_{\max})$ of the  rectangular obstacles as $\ell_{\min}=\ell_0<\ell_1<...<\ell_n=\ell_{\max}$. Let $\Delta \ell_j = \ell_j - \ell_{j-1}, j=1,...,n,$ be the length of the $j$-th subinterval $[\ell_{j-1}, \ell_j)$. Denote by
\begin{align}
\Phi_o (w_i, \ell_j) = \{ x \in  \Phi_o: W_x \in  [w_{i-1}, w_i),  L_x \in  [\ell_{j-1}, \ell_j) \}. \notag 
\end{align}
In other words, $\Phi_o (w_i, \ell_j)$ collects the centroid points of the obstacles whose grains' widths and lengths are in $[w_{i-1}, w_i) \times [\ell_{j-1}, \ell_j)$. Because points in the PPP $\Phi_o$ are independently marked by $\{\Xi(W_{x_i},L_{x_i})\}$, $\Phi_o (w_i, \ell_j)$ is an independent thinning of the PPP $\Phi_o$. The probability that a given point $x$ in $\Phi_o$ is retained equals $(  F_W (w_i) - F_W (w_{i-1})  ) \cdot ( F_L (\ell_j) - F_L (\ell_{j-1}) )$. It follows that $\Phi_o (w_i, \ell_j)$ is also a PPP with intensity 
\begin{align}
\lambda_o ( w_i, \ell_j ) = &\lambda_o   (  F_W (w_i) - F_W (w_{i-1})  )   ( F_L (\ell_j) - F_L (\ell_{j-1}) ).
\end{align}
Moreover, $\Phi_o (w_i, \ell_j), \forall i, j$, are independent.

Next we approximate all the obstacles whose centroid points are in $\Phi_o (w_i, \ell_j)$ by the same rectangle $\Xi (w_i, \ell_j)$. This approximation becomes exact as $\delta = \max ( \Delta w_i , i=1,...,m, \Delta \ell_j , j=1,...,n ) $ goes to zero (with increasing $m$ and $n$). With this approximation, the intensity of $\Phi_o (w_i, \ell_j)$ can be written as
\begin{align}
\lambda_o ( w_i, \ell_j ) = \lambda_o  f_W (w_i)  \Delta w_i  f_L (\ell_j)   \Delta \ell_j .
\end{align}
Denote by $B_{i,j}$ the event that the obstacles whose centroid points are located in $\Phi_o (w_i, \ell_j)$ cause blockage (i.e., no path connecting  the source $x$ to the destination $y$ can be found). It follows that
\begin{align}
\mathbb{P} (\textrm{Blockage}) \geq \mathbb{P} ( \cup_{i,j} B_{i,j} ) =1 - \prod_{i,j} \left( 1 - \mathbb{P} ( B_{i,j} ) \right).
\end{align}
where the inequality is due to the fact that we ignore that obstacles whose centroid points belong to different $\Phi_o (w_i, \ell_j)$'s may jointly cause blockage, and the last equality is due to the independence of $\Phi_o (w_i, \ell_j)$, $\forall i, j$.

To calculate a lower bound on the blockage probability, we further consider the blockage event $\tilde{B}_{i,j}$ caused by a single obstacle whose centroid is located in $\Phi_o (w_i, \ell_j)$.  Clearly, $\mathbb{P} ( B_{i,j} ) \geq \mathbb{P} ( \tilde{B}_{i,j} ) $ because obstacles whose centroid points belong to the same $\Phi_o (w_i, \ell_j)$ may also jointly cause blockage. To calculate $\mathbb{P} ( \tilde{B}_{i,j} ) $, we consider the following four cases, which are illustrated in Fig. \ref{fig:2}.

\textbf{Case 1:} $\ell_j < \kappa, w_i < d$. As $\ell_j < \kappa$, a single obstacle cannot ``cut off" the rectangular strip $R_\kappa (x,y)$. The blockage happens when either the source $x$ or the destination $y$ located in the obstacle. The last event happens when there exists at least one obstacle whose centroid point is in the region of the two red rectangles shown in the upper-left plot of Fig. \ref{fig:2}. Denoting by $A$ the area of the two red rectangles, 
\begin{align}
&\mathbb{P} ( \tilde{B}_{i,j} |  \ell_j < \kappa, w_i < d) = \mathbb{P} ( \exists x \in \Phi_o (w_i, \ell_j) \textrm{ s.t. } x \in A  )  \notag \\
&= 1 - e^{-\lambda_o ( w_i, \ell_j ) | A | } = 1 - e^{-\lambda_o ( w_i, \ell_j ) \cdot 2 w_i \ell_j }, 
\end{align}
where in the second equality we have used the fact that the number of obstacles whose centroid points are located in $A$ is Poisson distributed with mean $\lambda_o ( w_i, \ell_j ) | A |$.

\begin{figure}
\centering
\includegraphics[width=8cm]{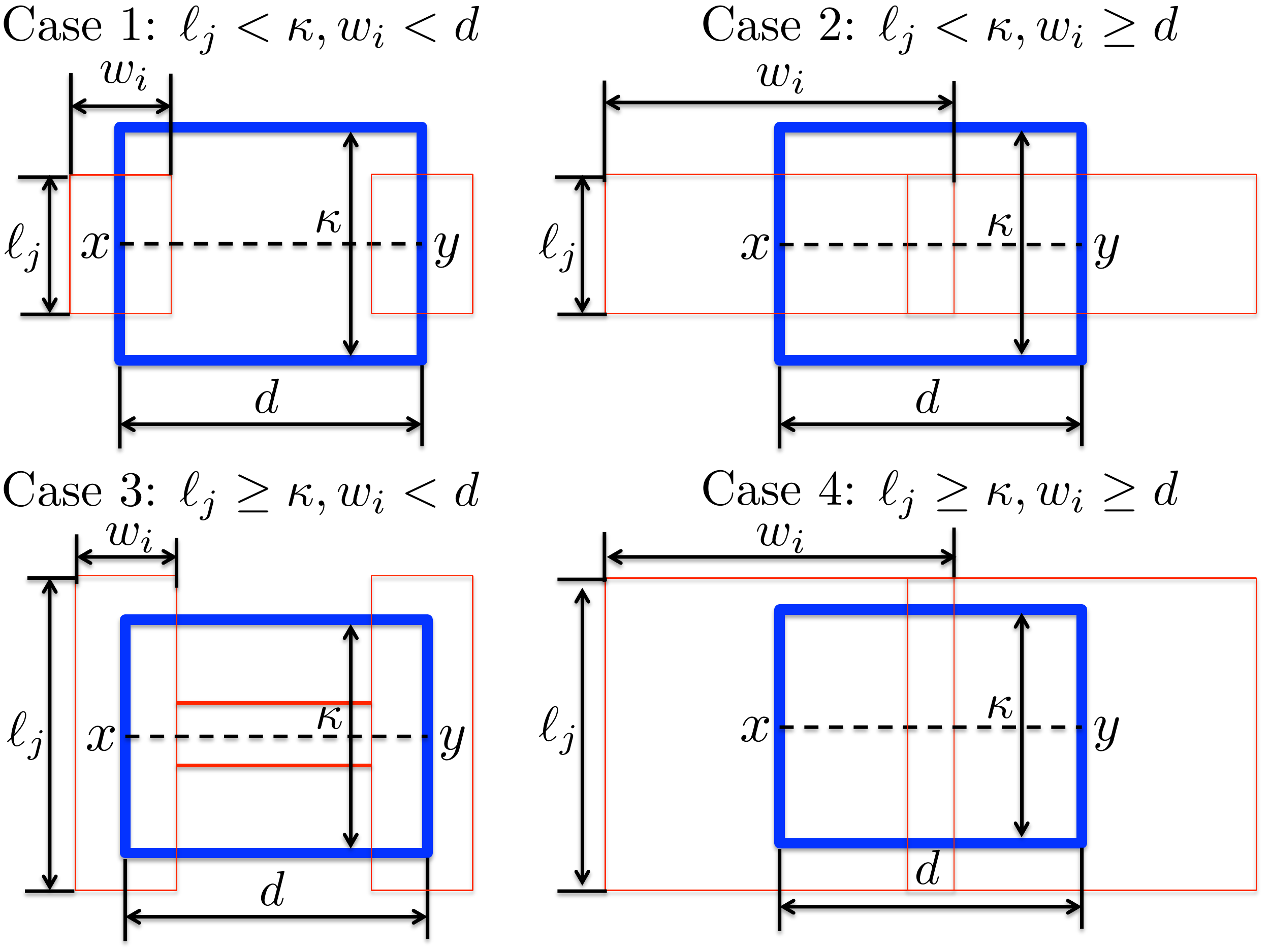}
\caption{Geometric illustration for the proof of Theorem \ref{thm:1}.}
\label{fig:2}
\end{figure}

\textbf{Case 2:}  $\ell_j < \kappa, w_i \geq d$. The analysis of this case is similar to Case 1. The difference is that the intersection of the two red rectangles now is not empty due to $w_i \geq d$. As shown in the upper-right plot of Fig. \ref{fig:2}, the area of the union of the two red rectangles equals $( w_i +d )\ell_j$ . It follows that
\begin{align}
\mathbb{P} ( \tilde{B}_{i,j} |  \ell_j < \kappa, w_i \geq d) 
&= 1 - e^{-\lambda_o ( w_i, \ell_j ) \cdot ( w_i +d )\ell_j }.
\end{align}

\textbf{Case 3:}  $\ell_j \geq \kappa, w_i < d$. As $\ell_j \geq \kappa$, a single obstacle can ``cut off" the rectangular strip $R_\kappa (x,y)$. It can be verified that if the centroid point of any obstacle is located in the red region shown in the bottom-left plot of Fig. \ref{fig:2}, the blockage occurs. The region has area $(\ell_j - \kappa) (d - w_i) + 2 w_i \ell_j$. It follows that 
\begin{align}
& \mathbb{P} ( \tilde{B}_{i,j} |  \ell_j \geq \kappa, w_i < d) 
=  1 - e^{-\lambda_o ( w_i, \ell_j ) \cdot (\ell_j - \kappa) (d - w_i) - \lambda_o ( w_i, \ell_j ) \cdot 2 w_i \ell_j  }.
\end{align}

\textbf{Case 4:}  $\ell_j \geq \kappa, w_i \geq d$. This case can be analyzed along the same line as in Cases 2 and 3, leading to
\begin{align}
\mathbb{P} ( \tilde{B}_{i,j} |  \ell_j < \kappa, w_i \geq d) 
&= 1 - e^{-\lambda_o ( w_i, \ell_j ) \cdot ( w_i +d )\ell_j }.
\end{align}

Combining the above four cases yields 
\begin{align}
\mathbb{P} (\textrm{Blockage}) &\geq 1 - \prod_{i,j} \left( 1 - \mathbb{P} ( \tilde{B}_{i,j} ) \right)  = 1 - e^{ - \sum_{i,j} \lambda_o ( w_i, \ell_j ) | A_{i,j} |  } \notag \\
&= 1 - e^{ - \sum_{i,j} \lambda_o  f_W (w_i)  \Delta w_i  f_L (\ell_j)   \Delta \ell_j | A_{i,j} |  } ,
\end{align}
where 
\[ | A_{i,j} | = \left\{ \begin{array}{ll}
         ( w_i +d )\ell_j & \mbox{if $w_i \geq d$};\\
         2 w_i  \ell_j & \mbox{if $w_i < d$ and $\ell_j < \kappa$};\\
         (\ell_j - \kappa) (d - w_i) + 2 w_i \ell_j& \mbox{if $w_i < d$ and $\ell_j \geq \kappa$}.\end{array} \right. \]

Finally, we keep refining the partitions of the width interval $[w_{\min}, w_{\max})$ and the length interval $[\ell_{\min}, \ell_{\max})$ such that 
$\delta = \max ( \Delta w_i , i=1,...,m, \Delta \ell_j , j=1,...,n )$ goes to zero. Further using the fact $\mathbb{P} (\mathcal{E}) = 1-\mathbb{P} (\textrm{Blockage})$ completes the proof.

\section{Conclusions}

This paper has analyzed the connectivity of mmWave networks with multi-hop relaying. An interesting result is that relaying route window can be set to be about the size of the obstacles to obtain near-optimal connectivity. Note that with a large relaying route window, it is more likely that the mmWave signals take on a long route, leading to more total energy consumption. Future work could explore such a tradeoff between connectivity and energy consumption.

\bibliographystyle{IEEEtran}
\bibliography{IEEEabrv,Reference}

\end{document}